\documentclass[11pt]{article}

\input{glyphtounicode}
\usepackage[authoryear,round,semicolon]{natbib}
\usepackage{makecell} %
\usepackage{booktabs}
\usepackage{amsmath,amsthm,amssymb} %
\usepackage[affil-it]{authblk} %
\usepackage[english]{babel} %
\usepackage{caption} %
\usepackage{color} %
\usepackage{algorithmic} %
\usepackage{algorithm} %
\usepackage{enumitem} %
\usepackage{mathpazo} 
\usepackage{framed} %
\usepackage[margin=1in]{geometry} %
\usepackage{graphics} %
\usepackage{hyphenat} %
\usepackage[breaklinks]{hyperref} %
\usepackage{mathtools} %
\usepackage{microtype} %
\usepackage[T1]{fontenc} %
\usepackage[utf8]{inputenc} %
\usepackage{soul} %
\usepackage{subcaption} %
\usepackage{subdepth} %
\usepackage{suffix} 
\usepackage{tikz} %
\usepackage{xspace} %
\usepackage{thmtools}
\usepackage{thm-restate}
\usepackage{interval}
\intervalconfig{soft open fences}
\usepackage{cleveref}

\hypersetup{colorlinks=true, linkcolor=blue, citecolor=magenta} %
\setcitestyle{aysep={,},yysep={;}}
\definecolor{shadecolor}{rgb}{.95,.95,.95} %

\setul{1ex}{.5pt} %
\setlist[enumerate]{nolistsep,itemsep=3pt,topsep=3pt} %

\definecolor{White}{rgb}{1,1,1} %
\definecolor{Black}{rgb}{0,0,0} %
\definecolor{LightGray}{rgb}{.8,.8,.8} %

\theoremstyle{plain}
\newtheorem{theorem}{Theorem} %
\newtheorem{lemma}[theorem]{Lemma} %
\theoremstyle{definition}
\newtheorem{conjecture}[theorem]{Conjecture} %
\theoremstyle{remark}
\newcommand{\RR}{\mathbb{R}} %
\newcommand{\trace}{\operatorname{Tr}} %
\renewcommand{\epsilon}{\varepsilon} %

\newcommand{\Top}{\operatorname{Top}}
\newcommand{\OPT}{\operatorname{OPT}}
\newcommand{\KL}{\operatorname{KL}}
\newcommand{\rank}{\operatorname{rank}}

\newif\ifanonymous
\anonymousfalse

\begin{document}

\title{\Large\bf A Correlation-Gap Bound for Nonlinear Gaussian PCA}
\date{}
\renewcommand*{\Affilfont}{\small\itshape} 
\ifanonymous
\author{Anonymous Authors}
\affil{}
\else
\author[1,2]{Minbo Gao}
\author[3]{Zhengfeng Ji}
\author[1,2]{Chenghua Liu}
\affil[1]{
  Institute of Software, Chinese Academy of Sciences, Beijing, China
}
\affil[2]{
  University of Chinese Academy of Sciences, Beijing, China
}
\affil[3]{
Department of Computer Science and Technology, Tsinghua University, Beijing, China
}
\fi

\maketitle

\begin{abstract}
Principal component analysis (PCA) is optimal for the linear reconstruction of Gaussian data, a foundational property underlying its central role in algorithms and signal processing. Its nonlinear analogue, however, is notoriously subtle: in 2011, Mallat and Zeitouni conjectured that the Karhunen--Lo\`eve (KL) basis remains optimal even when the retained coordinates are chosen adaptively per sample, a property that would theoretically justify the ubiquitous pipeline of PCA followed by sparse thresholding. In this paper, we establish a $1+O(1/\sqrt{d})$-approximate version of the retained-energy form of the Mallat--Zeitouni conjecture, showing that the KL basis is within this factor of the optimal basis.
This dimension-free comparison depends only on the number of retained coordinates and shows that the possible advantage of optimizing over all orthonormal bases vanishes as $d$ grows. It complements the universal-constant reconstruction-error comparison of Litvak and Tikhomirov (\emph{Ann.\ Appl.\ Probab.}, 2018), while  providing a comparison naturally suited for algorithmic analysis.
Our proof rests on a clean, conceptual reduction: we relax arbitrary rotations to a deterministic threshold bound via Schur--Horn majorization, and identify the remaining loss with the correlation gap of the rank-$d$ uniform matroid over Gaussian level sets.
\end{abstract}

\newpage

\section{Introduction}

Principal component analysis (PCA) is a foundational spectral primitive for linear dimension reduction.  Given second-order information about a distribution, PCA chooses the fixed low-dimensional subspace that captures the most variance, or equivalently minimizes mean-square reconstruction error.  This exact optimality principle, originating in the classical works of \citet{Pearson1901} and \citet{Hotelling1933}, is the statistical counterpart of the Eckart--Young--Mirsky theorem for low-rank approximation \citep{EckartYoung1936,Mirsky1960,Jolliffe2002}.  It has made PCA a canonical benchmark across statistics, signal processing, numerical linear algebra, and algorithms.  Modern randomized numerical linear algebra has further turned PCA into a scalable algorithmic workhorse: randomized SVD and sketching give fast low-rank approximations \citep{HalkoMartinssonTropp2011,Woodruff2014}, while subspace iteration and block Krylov methods provide stronger guarantees for approximate principal components \citep{MuscoMusco2015,MartinssonTropp2020}.

The classical PCA guarantee is linear and non-adaptive: the subspace is chosen once, before seeing the sample, and every sample is projected onto that same subspace.  Many approximation and compression procedures use a more adaptive rule.  After expanding a signal in a basis, they keep the largest \(d\) coefficients of the particular realization and discard the rest.  This best \(d\)-term rule is a central object in nonlinear approximation \citep{DeVore1998,Temlyakov2008,Temlyakov2011}; in wavelet transform coding, large coefficients encode the dominant structure of a signal or image \citep{DeVoreJawerthLucier1992,Mallat2009}; and in sparse representation and recovery algorithms, thresholding or greedy selection is a basic mechanism for constructing sparse approximants \citep{MallatZhang1993,ChenDonohoSaunders2001,BlumensathDavies2009}.  The rule is nonlinear because the retained coordinates depend on the observed sample.  It is therefore a stronger benchmark than fixed-subspace projection: rather than asking for the best fixed \(d\)-dimensional approximation space, it asks for a coordinate system in which sample-dependent top-\(d\) selection captures as much signal energy as possible.

For Gaussian data, the natural candidate coordinate system is the Karhunen--Lo\`eve basis, namely the eigenbasis of the covariance matrix.  The Mallat--Zeitouni conjecture \citep{MallatZeitouni2011} asks whether this basis remains optimal even for the nonlinear best \(d\)-term approximation rule.  More concretely, their conjecture can be written as follows:

\begin{conjecture}[{\citet[Conjecture 1]{MallatZeitouni2011}}]
\label{conj:mz}
Let \(X=(X_1,\ldots,X_p)\) be a centered Gaussian vector with independent
coordinates and variances \(\lambda_1\ge\cdots\ge\lambda_p\).  Let
\(z_{[1]}^2\ge\cdots\ge z_{[p]}^2\) denote the decreasing rearrangement
of the squared coordinates of a vector \(z\). For every \(U\in O(p)\) and
\(1\le d<p\), we have
\begin{equation}
\label{eq:mz}
\mathbb E\sum_{j=d+1}^{p}X_{[j]}^2
\le
\mathbb E\sum_{j=d+1}^{p}(UX)_{[j]}^2 .
\end{equation}
\end{conjecture}

From an algorithmic viewpoint, the conjecture isolates a basic quantifier switch.  In classical PCA, one chooses a subspace before seeing the sample.  In nonlinear approximation, one chooses the basis first, then after seeing the sample selects the best \(d\) coordinates of that sample.  If the coordinate subset also had to be fixed before sampling, the top eigenspace would be optimal by the usual eigenvalue-majorization argument.  The Mallat--Zeitouni problem asks whether this spectral optimality survives when the subset is allowed to depend on the sample.  \citet{Bandeira2016} formulated the same question as a data-science open problem: choose a basis before observing \(g\sim N(0,\Sigma)\), and after observing \(g\) choose the \(d\) basis vectors whose span captures the most energy.

This viewpoint connects the conjecture to a basic algorithmic pattern: choose a representation from coarse statistical information, and then sparsify adaptively after seeing the input.  The second step is the same sample-dependent support selection that appears throughout sparse approximation and recovery.  For example, CoSaMP combines greedy support identification with pruning \citep{NeedellTropp2009}, while iterative hard thresholding makes the hard-thresholding step explicit \citep{BlumensathDavies2009}.  Sketching-based sparse recovery also measures performance against the best \(s\)-term approximation of the input, seeking fast linear measurements and decoders that compete with this adaptive benchmark \citep{GilbertIndyk2010,GilbertLiPoratStrauss2012}.  A related line of work treats the representation itself as something to be learned: provable dictionary learning asks when a sparsifying dictionary can be recovered from data \citep{SpielmanWangWright2012,AroraGeMoitra2014}.  The Mallat--Zeitouni problem strips these themes down to a particularly clean model: Gaussian data, known covariance, and orthonormal dictionaries.  It asks whether any rotation can improve over covariance diagonalization before the adaptive top-\(d\) selection step.  A positive answer would certify PCA as the optimal preprocessing not only for fixed linear projections, but also for this basic nonlinear postprocessing rule; a negative answer would show that sample-dependent sparse approximation can exploit rotations invisible to classical PCA.

\paragraph{Prior work and obstacles.}
A first obstacle is that nonlinear approximation can be much more powerful than its linear counterpart.  For any fixed basis, the best \(d\)-term error is at most the error of projecting onto any fixed \(d\) coordinates.  The gap can be substantial: \citet{LitvakTikhomirov2018} give a Gaussian example in which the linear \(d\)-term error is of order \(1\), while the nonlinear \(d\)-term error is of order \(d^{-2}\).  Thus \eqref{eq:mz} is not a formal extension of the PCA variational principle.  After an orthogonal rotation, the coordinates of \(UX\) are simultaneously dependent and non-identically distributed, and the objective is a sum of order statistics of their squares.  Classical order-statistic theory primarily treats independent samples, with extensions to independent non-identically distributed variables and to special dependence structures such as exchangeability or equicorrelation \citep{DavidNagaraja2003}.

\citet{MallatZeitouni2011} proved the conjecture in the case \(d=1\), namely reconstruction from the single largest projection.  Their proof combines a \v{S}id\'ak-type Gaussian correlation inequality with a Marshall--Proschan majorization argument.  The majorization part applies naturally to the sum of the largest \(d\) squared coordinates, but the required decoupling step fails.  More precisely, one would like to replace a dependent Gaussian vector by an independent Gaussian vector with the same coordinate variances and only increase the expected top-\(d\) retained energy.  Ramon van Handel communicated a three-dimensional counterexample, included in \citet{MallatZeitouni2011}, showing that this intermediate assertion is false already for \(p=3\) and \(d=2\).  Thus the conjecture cannot be proved by first discarding the Gaussian dependence structure and replacing the rotated vector by independent coordinates with the same marginal variances.

The main progress toward the original reconstruction-error conjecture is due to \citet{LitvakTikhomirov2018}, who proved \eqref{eq:mz} up to a universal multiplicative constant.  Their proof introduces a comparison theorem for sums of order statistics: it compares a vector with independent coordinates to a vector with arbitrary dependence under mild one-dimensional distributional assumptions.  This result is of independent interest and places the Mallat--Zeitouni problem in a broader line of work on Gaussian minima, dependent order statistics, and Orlicz-norm methods \citep{GordonLitvakSchuettWerner2005,GordonLitvakSchuettWerner2006,GordonLitvakSchuettWerner2012}.  The exact constant-one conjecture remains open; it also continues to motivate related extremal questions for Gaussian order statistics \citep{Bandeira2016,Litvak2018Simplex,Kunisky2026}.

\paragraph{Our result and proof sketch.}
Since orthogonal matrices preserve the Euclidean \(\ell_2\) norm of a vector,
\(\|UX\|_2=\|X\|_2\), Conjecture~\ref{conj:mz} can be equivalently
written as the retained-energy form:
\[
    \mathbb E\sum_{j=1}^{d}(UX)_{[j]}^2
    \le
    \mathbb E\sum_{j=1}^{d}X_{[j]}^2 .
\]
Our main theorem, \Cref{thm:main}, proves the following
\(1+O(d^{-1/2})\)-approximate version:
\[
    \mathbb E\sum_{j=1}^{d}(UX)_{[j]}^2
    \le
    \bigl(1+O(d^{-1/2})\bigr)
    \mathbb E\sum_{j=1}^{d}X_{[j]}^2 .
\]


The proof has two ingredients.  First, we upper bound the performance of every rotated basis by a deterministic threshold relaxation.  The sum of the largest \(d\) coordinates admits an exact threshold variational formula; moving the threshold outside the expectation gives a relaxation depending only on the coordinate variances in the chosen basis.  By the Schur--Horn theorem, this variance vector is majorized by the eigenvalue vector of \(\Sigma\), and by Karamata's inequality the resulting convex relaxation is maximized in the Karhunen--Lo\`eve basis.

The second step exposes the combinatorial structure hidden in adaptive top-\(d\) selection.  The uniform-matroid structure is already implicit in the adaptive top-\(d\) selection rule: at any fixed level, one may retain at most \(d\) coordinates above that level.  After a layer-cake decomposition over Gaussian level sets, the retained-energy objective in the Karhunen--Lo\`eve basis becomes an integral of the rank function \(S\mapsto \min(d,|S|)\) of the rank-\(d\) uniform matroid, evaluated on product Bernoulli random sets.  The deterministic threshold relaxation corresponds, level by level, to the largest value allowed by the same marginal probabilities, namely \(\min(d,\sum_i q_i)\).  Thus, after the threshold relaxation, the quantitative loss in our argument is precisely the uniform-matroid correlation gap, equivalently the balancedness of the corresponding contention-resolution scheme.

This connection brings tools from stochastic optimization and submodular rounding into the Mallat--Zeitouni problem.  Correlation gaps were introduced in correlation-robust stochastic optimization \citep{AgrawalDingSaberiYe2009} and used in mechanism design \citep{Yan2011}; contention-resolution schemes were developed as a rounding framework for submodular optimization \citep{ChekuriVondrakZenklusen2014}.  We use the sharp offline uniform-matroid constant of \citet{KashaevSantiago2023}.  Closely related online contention-resolution schemes for uniform matroids, motivated by prophet inequalities, were recently studied by \citet{DinevWeinberg2024}; our argument uses the offline correlation-gap form.  Recent work of \citet{HusicKohLohoVegh2022} develops a more fine-grained theory of matroid correlation gaps.

\paragraph{Discussion.}
Our result gives an explicit, dimension-free retained-energy comparison for the Karhunen--Lo\`eve basis under adaptive top-\(d\) truncation.  At the exact constant-one level this is equivalent to the original reconstruction-error conjecture, since the two objectives sum to \(\trace(\Sigma)\); for multiplicative approximations, however, the two formulations are not interchangeable.  Thus our result complements the universal-constant reconstruction-error theorem of \citet{LitvakTikhomirov2018}, giving a near-one guarantee for the retained-energy objective.  The argument reduces arbitrary rotations to their marginal variance vectors via Schur--Horn majorization, and then compares the resulting Gaussian level sets through the rank-\(d\) uniform-matroid correlation gap.  Within this marginal-variance relaxation, this gap is sharp.  This suggests that a full resolution of the original Mallat--Zeitouni conjecture will likely require ideas beyond this relaxation, in particular a way to capture how the joint Gaussian dependence created by rotation interacts with adaptive top-\(d\) selection.

\section{Preliminaries}
We write \([p]\coloneqq\{1,\ldots,p\}\) and
\(x_+\coloneqq\max\{x,0\}\).  
For
\(a=(a_1,\ldots,a_p)\in\RR_+^p\), let
$a_{[1]}\ge a_{[2]}\ge\cdots\ge a_{[p]}$
denote the non-increasing rearrangement of its coordinates.  For
\(1\le d\le p\), define
$\Top_d(a)\coloneqq \sum_{j=1}^d a_{[j]}$.
Equivalently,
    $\Top_d(a)=
    \max_{{S\subseteq[p], |S|=d}}
    \sum_{i\in S} a_i.$

Let \(g\sim N(0,\Sigma)\) in \(\RR^p\), where \(\Sigma\succeq0\).  An
orthonormal basis is written as \(V=(v_1,\ldots,v_p)\in O(p)\), with
\(v_i\) as columns.  Its nonlinear retained \(d\)-term energy is
\[
    \mathcal R_d(V;\Sigma)
    \coloneqq
    \mathbb E\,
    \Top_d\bigl((v_1^\top g)^2,\ldots,(v_p^\top g)^2\bigr).
\]
The best retained energy over all orthonormal bases is
\[
    \OPT_d(\Sigma)\coloneqq
    \sup_{V\in O(p)} \mathcal R_d(V;\Sigma).
\]
Intuitively, the optimization has two stages: the basis \(V\) is chosen before
the sample is observed, and the top \(d\) coordinates are chosen afterward.

Let $\lambda_1\ge\lambda_2\ge\cdots\ge\lambda_p\ge0$
be the eigenvalues of \(\Sigma\), including zero eigenvalues, and let
$ r\coloneqq \rank(\Sigma)=|\{i:\lambda_i>0\}|.$
Throughout, \(Z,Z_1,\ldots,Z_p\) denote independent standard Gaussian random
variables.  In a Karhunen--Lo\`eve eigenbasis, \(g\) has coordinates
    $(\sqrt{\lambda_1}Z_1,\ldots,\sqrt{\lambda_p}Z_p)$,
and we write
\[
    \KL_d(\Sigma)
    \coloneqq
    \mathbb E\,\Top_d
    \bigl(\lambda_1Z_1^2,\ldots,\lambda_pZ_p^2\bigr).
\]
Since a Karhunen--Lo\`eve eigenbasis is feasible in the definition of
\(\OPT_d(\Sigma)\), we always have
    $\KL_d(\Sigma)\le \OPT_d(\Sigma)$.

We use standard majorization notation.  For \(x,y\in\RR^p\), write
\(x\preceq y\) if, after sorting both vectors in non-increasing order,
\[
\sum_{i=1}^{\ell}x_{[i]}
\le
\sum_{i=1}^{\ell}y_{[i]},
\qquad 1\le \ell<p,
    \quad\text{and}\quad
    \sum_{i=1}^p x_i=\sum_{i=1}^p y_i .
\]
For \(1\le d\le r\), let \(\rho_d\) denote the rank function of the
rank-\(d\) uniform matroid on \([r]\):
\[
    \rho_d(S)\coloneqq \min(d,|S|),
    \qquad S\subseteq[r].
\]

\section{Main Theorem and Its Proof}


\begin{theorem}[Correlation-gap bound for nonlinear Gaussian PCA]
\label{thm:main}
Let \(1\le d\le p\), and let \(g\sim N(0,\Sigma)\) in \(\RR^p\).  Let
\(\lambda_1\ge\cdots\ge\lambda_p\ge0\) be the eigenvalues of \(\Sigma\), and
let \(r=\rank(\Sigma)\).  If \(r\le d\), then
\[
    \OPT_d(\Sigma)=\KL_d(\Sigma)=\trace(\Sigma).
\]
If \(r>d\), then
\[
    \KL_d(\Sigma)
    \le
    \OPT_d(\Sigma)
    \le
    c_{d,r}^{-1}\KL_d(\Sigma)
    \le
    \gamma_d^{-1}\KL_d(\Sigma),
\]
where
\[
    c_{d,r}
    \coloneqq
    1-
    \binom r d
    \left(\frac dr\right)^d
    \left(1-\frac dr\right)^{r+1-d},
\]
and $\gamma_d\coloneqq 1-e^{-d}\frac{d^d}{d!}$.
Moreover,
\[
    \gamma_d^{-1}
    =
    1+\frac{1}{\sqrt{2\pi d}}+O(d^{-1})
    =
    1+O(d^{-1/2}).
\]
\end{theorem}

The proof separates the analytic and combinatorial parts of the problem.
First, every rotated basis is upper bounded by a deterministic threshold
relaxation that depends only on the diagonal variances of that basis.  The
Schur--Horn theorem and Karamata's inequality then show that this relaxation
is maximized by the eigenvalue vector.  Second, in the Karhunen--Lo\`eve
basis, a layer-cake decomposition rewrites the retained energy as an integral
of the rank function \(\rho_d\) over independent Bernoulli level sets.  The
loss between the threshold relaxation and the independent Gaussian level sets
is exactly the rank-\(d\) uniform-matroid correlation gap.

For the convenience of the proof, for \(\lambda\in\RR_+^p\), we define the threshold relaxation 
\[
    \mathsf U_d(\lambda)
    \coloneqq
    \inf_{\tau\ge0}
    \left\{
        d\tau+
        \sum_{i=1}^p
        \mathbb E(\lambda_iZ_i^2-\tau)_+
    \right\}.
\]

\subsection{A deterministic threshold relaxation}

We begin with a pointwise identity for the sum of the largest \(d\)
coordinates.  It is the analytic step that lets us replace the adaptive
top-\(d\) choice by a one-dimensional threshold.

\begin{lemma}[Top-\(d\) threshold identity]
\label{lem:threshold-identity}
For every \(a\in\RR_+^p\) and every \(1\le d\le p\),
\[
    \Top_d(a)
    =
    \inf_{\tau\ge0}
    \left\{
        d\tau+\sum_{i=1}^p(a_i-\tau)_+
    \right\}.
\]
\end{lemma}

\begin{proof}
Let
\[
    F(\tau)\coloneqq
    d\tau+\sum_{i=1}^p(a_i-\tau)_+ .
\]
For every \(\tau\ge0\),
\[
    F(\tau)
    \ge
    d\tau+\sum_{j=1}^d(a_{[j]}-\tau)
    =
    \Top_d(a).
\]
Conversely, put \(a_{[p+1]}=0\) and choose
\[
    \tau\in [a_{[d+1]},a_{[d]}],
\]
with \(\tau=0\) when \(d=p\).  Then
\[
    (a_{[j]}-\tau)_+=a_{[j]}-\tau
    \quad\text{for }j\le d,
    \qquad
    (a_{[j]}-\tau)_+=0
    \quad\text{for }j>d.
\]
Hence \(F(\tau)=\Top_d(a)\).  This proves the identity.
\end{proof}

\begin{lemma}[Threshold relaxation dominates every basis]
\label{lem:threshold-relaxation-dominates}
For every orthonormal basis \(V=(v_1,\ldots,v_p)\),
\[
    \mathcal R_d(V;\Sigma)
    \le
    \mathsf U_d(\lambda).
\]
Consequently,
\[
    \OPT_d(\Sigma)\le \mathsf U_d(\lambda).
\]
\end{lemma}

\begin{proof}
Fix \(V=(v_1,\ldots,v_p)\in O(p)\), and write
\[
    Y_i=v_i^\top g,
    \qquad
    \mu_i=\operatorname{Var}(Y_i)=v_i^\top\Sigma v_i .
\]
By Lemma~\ref{lem:threshold-identity}, for every \(\tau\ge0\),
\[
    \Top_d(Y_1^2,\ldots,Y_p^2)
    \le
    d\tau+\sum_{i=1}^p (Y_i^2-\tau)_+ .
\]
Taking expectations and then minimizing over \(\tau\) gives
\[
    \mathcal R_d(V;\Sigma)
    \le
    \inf_{\tau\ge0}
    \left\{
        d\tau+\sum_{i=1}^p
        \mathbb E(Y_i^2-\tau)_+
    \right\}.
\]
Since \(Y_i\sim N(0,\mu_i)\),
\[
    \mathbb E(Y_i^2-\tau)_+
    =
    h_\tau(\mu_i),
    \qquad
    h_\tau(u)\coloneqq \mathbb E(uZ^2-\tau)_+ .
\]
For fixed \(\tau\), the function \(h_\tau:\RR_+\to\RR_+\) is convex, because
\(u\mapsto (uZ^2-\tau)_+\) is convex for every value of \(Z\).  The vector
\(\mu=(\mu_1,\ldots,\mu_p)\) is the diagonal of \(V^\top\Sigma V\), and hence
\(\mu\preceq\lambda\) by Schur's Theorem (see~\citet[Exercise II.1.12]{bhatiaMatrixAnalysis1997}).  Karamata's inequality (see~\citet[Theorem II.3.1]{bhatiaMatrixAnalysis1997}) yields
\[
    \sum_{i=1}^p h_\tau(\mu_i)
    \le
    \sum_{i=1}^p h_\tau(\lambda_i)
    \qquad
    \text{for every }\tau\ge0.
\]
Therefore
\[
\begin{aligned}
    \mathcal R_d(V;\Sigma)
    &\le
    \inf_{\tau\ge0}
    \left\{
        d\tau+\sum_{i=1}^p h_\tau(\mu_i)
    \right\}                                                \\
    &\le
    \inf_{\tau\ge0}
    \left\{
        d\tau+\sum_{i=1}^p h_\tau(\lambda_i)
    \right\}
    =
    \mathsf U_d(\lambda).
\end{aligned}
\]
Taking the supremum over \(V\in O(p)\) proves the consequence.
\end{proof}

\subsection{Gaussian layers and the uniform-matroid correlation gap}

The preceding relaxation has a clean interpretation.  At each level \(t\),
the Karhunen--Lo\`eve basis produces an independent random set of coordinates
whose squared values exceed \(t\).  The retained-energy objective integrates
the rank function \(\rho_d(S)=\min(d,|S|)\) over these level sets.  The
threshold relaxation instead integrates the largest value that any coupling
with the same one-dimensional marginals could have, namely
\(\min(d,\sum_i q_i)\).  The comparison is therefore precisely a correlation
gap for the rank-\(d\) uniform matroid.

We use the following sharp offline form.

\begin{lemma}[Uniform-matroid Bernoulli bound]
\label{lem:uniform-matroid-bernoulli}
Let \(I_1,\ldots,I_r\) be independent Bernoulli random variables with
\(\mathbb P(I_i=1)=q_i\), and let \(N=\sum_{i=1}^r I_i\).  If \(1\le d<r\),
then
\[
    \mathbb E\min(d,N)
    \ge
    c_{d,r}
    \min\left(d,\sum_{i=1}^r q_i\right).
\]
\end{lemma}

\begin{proof}
First assume that \(\sum_i q_i\le d\).  Let \(R(q)\subseteq[r]\) be the
random set obtained by including \(i\) independently with probability \(q_i\).
Then \(q\) lies in the matroid polytope of the rank-\(d\) uniform matroid
\(U_d^r\).  By the optimal \(c_{d,r}\)-balanced contention-resolution scheme
of \citet[Theorem 2.1]{KashaevSantiago2023}, there is a possibly randomized rule
\(\pi_q\) such that
\[
    \pi_q(R(q))\subseteq R(q),
    \qquad
    |\pi_q(R(q))|\le d,
\]
and
\[
    \mathbb P(i\in \pi_q(R(q))\mid i\in R(q))
    \ge c_{d,r}
    \qquad
    \text{for every }i\in[r].
\]
Hence
\[
\begin{aligned}
    \mathbb E\min(d,|R(q)|)
    &\ge
    \mathbb E|\pi_q(R(q))|                                      \\
    &=
    \sum_{i=1}^r
    \mathbb P(i\in \pi_q(R(q)))                                  \\
    &\ge
    c_{d,r}\sum_{i=1}^r q_i .
\end{aligned}
\]

Now suppose \(\sum_iq_i>d\).  Choose \(q_i'\le q_i\) such that
\(\sum_iq_i'=d\).  Couple the two Bernoulli vectors by independent uniforms
\(U_i\sim\operatorname{Unif}[0,1]\):
\[
    I_i'=\mathbf 1_{\{U_i\le q_i'\}},
    \qquad
    I_i=\mathbf 1_{\{U_i\le q_i\}} .
\]
Then \(I_i'\le I_i\) almost surely.  Applying the first case to \(q'\),
\[
    \mathbb E\min\left(d,\sum_i I_i\right)
    \ge
    \mathbb E\min\left(d,\sum_i I_i'\right)
    \ge
    c_{d,r}d .
\]
Combining the two cases proves the lemma.
\end{proof}

\begin{lemma}[Layer-cake comparison for the KL basis]
\label{lem:layer-cake-kl}
If \(r>d\), then
\[
    \KL_d(\Sigma)
    \ge
    c_{d,r}\,\mathsf U_d(\lambda).
\]
\end{lemma}

\begin{proof}
The zero eigenvalues contribute neither to \(\KL_d(\Sigma)\) nor to
\(\mathsf U_d(\lambda)\), so we restrict attention to
\(i=1,\ldots,r\).  Put
\[
    B_i=\lambda_i Z_i^2 .
\]
For \(t\ge0\), define
\[
    q_i(t)=\mathbb P(B_i>t),
    \qquad
    m(t)=\sum_{i=1}^r q_i(t),
    \qquad
    N_t=\#\{i:B_i>t\}.
\]
For deterministic \(b=(b_1,\ldots,b_r)\in\RR_+^r\),
\[
    \Top_d(b)=
    \int_0^\infty
    \min\bigl(d,\#\{i:b_i>t\}\bigr)\,dt .
\]
Applying this identity to \(B=(B_1,\ldots,B_r)\) and using Tonelli's theorem,
\[
    \KL_d(\Sigma)
    =
    \int_0^\infty
    \mathbb E\min(d,N_t)\,dt .
\]
For each fixed \(t\), \(N_t\) is a sum of independent Bernoulli random
variables with means \(q_1(t),\ldots,q_r(t)\).  Lemma~\ref{lem:uniform-matroid-bernoulli}
therefore gives
\[
    \mathbb E\min(d,N_t)
    \ge
    c_{d,r}\min(d,m(t)).
\]
Thus
\[
    \KL_d(\Sigma)
    \ge
    c_{d,r}
    \int_0^\infty \min(d,m(t))\,dt .
\]

It remains to identify the last integral with the threshold relaxation.
Since
\[
    \mathbb E(B_i-\tau)_+
    =
    \int_\tau^\infty \mathbb P(B_i>t)\,dt,
\]
we have
\[
    \mathsf U_d(\lambda)
    =
    \inf_{\tau\ge0}
    \left\{
        d\tau+
        \int_\tau^\infty m(t)\,dt
    \right\}.
\]
The function \(m\) is continuous and non-increasing, with
\(m(0)=r>d\) and \(m(t)\to0\) as \(t\to\infty\).  Hence there exists
\(\tau_*\ge0\) such that \(m(\tau_*)=d\).  For \(\tau<\tau_*\),
\[
\begin{aligned}
    d\tau+\int_\tau^\infty m(t)\,dt
    -
    \left(
        d\tau_*+\int_{\tau_*}^\infty m(t)\,dt
    \right)
    &=
    \int_\tau^{\tau_*} (m(t)-d)\,dt
    \ge0,
\end{aligned}
\]
and for \(\tau>\tau_*\),
\[
\begin{aligned}
    d\tau+\int_\tau^\infty m(t)\,dt
    -
    \left(
        d\tau_*+\int_{\tau_*}^\infty m(t)\,dt
    \right)
    &=
    \int_{\tau_*}^{\tau} (d-m(t))\,dt
    \ge0.
\end{aligned}
\]
Therefore the infimum is attained at \(\tau_*\), and
\[
    \mathsf U_d(\lambda)
    =
    d\tau_*+\int_{\tau_*}^\infty m(t)\,dt
    =
    \int_0^\infty \min(d,m(t))\,dt .
\]
Combining this identity with the preceding lower bound on \(\KL_d(\Sigma)\)
proves the lemma.
\end{proof}
\subsection{Proof of the theorem and size of the constant}
\label{subsec:proof-main}

\begin{proof}[Proof of Theorem~\ref{thm:main}]
If \(r\le d\), then the Karhunen--Lo\`eve basis retains all nonzero
coordinates, and hence
\(\KL_d(\Sigma)=\mathbb E\|g\|_2^2=\trace(\Sigma)\).  On the other hand, for
every orthonormal basis \(V=(v_1,\ldots,v_p)\), pointwise
\[
    \Top_d\bigl((v_1^\top g)^2,\ldots,(v_p^\top g)^2\bigr)
    \le
    \sum_{i=1}^p (v_i^\top g)^2
    =
    \|g\|_2^2 .
\]
Taking expectations gives \(\OPT_d(\Sigma)\le\trace(\Sigma)\).  Since
\(\KL_d(\Sigma)\le\OPT_d(\Sigma)\), the equality
\(\OPT_d(\Sigma)=\KL_d(\Sigma)=\trace(\Sigma)\) follows.

Assume now that \(r>d\).  The Karhunen--Lo\`eve basis is feasible in the
definition of \(\OPT_d(\Sigma)\), so \(\KL_d(\Sigma)\le\OPT_d(\Sigma)\).
Lemma~\ref{lem:threshold-relaxation-dominates} gives
\(\OPT_d(\Sigma)\le\mathsf U_d(\lambda)\), while
Lemma~\ref{lem:layer-cake-kl} gives
\(\mathsf U_d(\lambda)\le c_{d,r}^{-1}\KL_d(\Sigma)\).  Therefore
\[
    \KL_d(\Sigma)
    \le
    \OPT_d(\Sigma)
    \le
    \mathsf U_d(\lambda)
    \le
    c_{d,r}^{-1}\KL_d(\Sigma).
\]
It remains only to pass from the finite-rank constant \(c_{d,r}\) to the
dimension-free constant \(\gamma_d\).  By the analysis of
\citet{KashaevSantiago2023}, the uniform-matroid balancedness constant
\(c_{d,r}\) converges from above, as \(r\to\infty\) with \(d\) fixed, to
\(\gamma_d=1-e^{-d}d^d/d!\).  Hence \(c_{d,r}\ge\gamma_d\), and so
\(c_{d,r}^{-1}\le\gamma_d^{-1}\).  Finally, Stirling's formula gives
\[
    \gamma_d^{-1}
    =
    \left(1-e^{-d}\frac{d^d}{d!}\right)^{-1}
    =
    1+\frac{1}{\sqrt{2\pi d}}+O(d^{-1}).
\]
This proves the theorem.
\end{proof}

\begin{figure}[t]
    \centering
    \resizebox{0.72\linewidth}{!}{\includegraphics{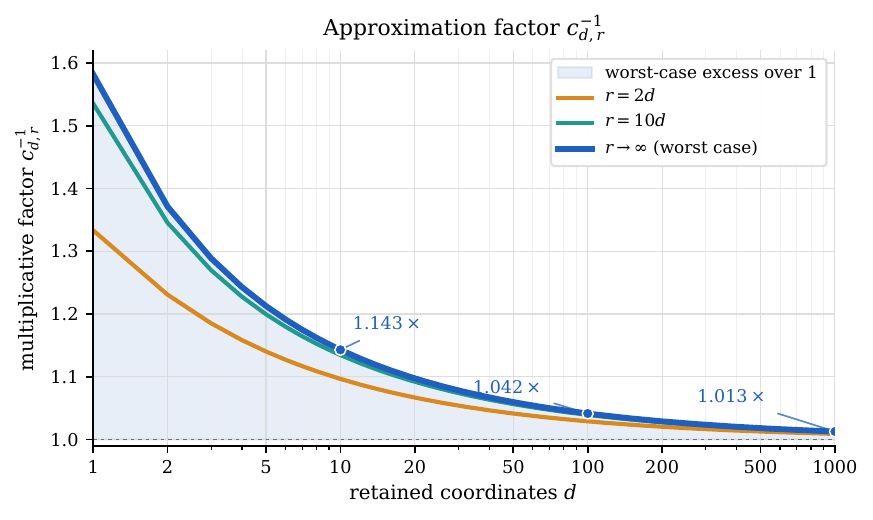}}
    \caption{
    The finite-rank and worst-case correlation-gap factors.  The orange and
    green curves show \(c_{d,2d}^{-1}\) and \(c_{d,10d}^{-1}\), while the
    blue curve is the worst-case limit \(\gamma_d^{-1}\).  The limiting factor
    is already about \(1.143\) at \(d=10\), \(1.042\) at \(d=100\), and
    \(1.013\) at \(d=1000\).
    }
    \label{fig:correlation-gap-constant}
\end{figure}

\paragraph{Quantitative interpretation.}
\Cref{fig:correlation-gap-constant} plots the finite-rank factors
\(c_{d,r}^{-1}\) and their worst-case limit \(\gamma_d^{-1}\). 
This is a
retained-energy guarantee, not a multiplicative reconstruction-error comparison
with \citet{LitvakTikhomirov2018}, since such bounds are not preserved after
subtracting from \(\trace(\Sigma)\).  

\bibliographystyle{elsarticle-harv}
\bibliography{ref}

\end{document}